\documentclass[a4paper]{llncs}
 
\usepackage[T1]{fontenc}
\usepackage[utf8]{inputenc}
\usepackage{url}

\usepackage{microtype}
\usepackage{nicefrac}
\usepackage{subfig}

\usepackage{comment}
\excludecomment{withproofs}

\usepackage{tikz}
\usetikzlibrary{arrows,shapes}

\usepackage{pgfplots}
\usepgfplotslibrary{external} 
\tikzset{
  tw-white/.style={draw,circle,fill=white,inner sep=2},
  tw-black/.style={draw,circle,fill=black,inner sep=2},
  tw-empty/.style={draw,circle,fill=black,inner sep=0.2},
  hav-white/.style={draw,circle,fill=white,inner sep=2},
  hav-black/.style={draw,circle,fill=black,inner sep=2},
  hav-empty/.style={regular polygon,regular polygon sides=6,draw,inner sep=5},
  hav-white-c/.style={draw,dashed,circle,fill=white,inner sep=2},
  hav-black-c/.style={draw,circle,fill=gray,inner sep=2},
}

\newcommand{\keywords}[1]{\par\addvspace\baselineskip
\noindent\keywordname\enspace\ignorespaces#1}

\newcommand{\gamename}{\textsc}

\newcommand{\np}{\textsc{np}}
\newcommand{\pspace}{\textsc{pspace}}
\newcommand{\GG}{\textsc{gg}}
\newcommand{\Hav}{\textsc{havannah}}

\title{Havannah and TwixT are \textsc{pspace}-complete}

\author{Édouard Bonnet\and Florian Jamain\and Abdallah Saffidine
}
\institute{L\textsc{amsade}, Université Paris-Dauphine, \url{{edouard.bonnet, abdallah.saffidine}@dauphine.fr}\\\url{florian.jamain@lamsade.dauphine.fr}}

\begin{document}

\mainmatter  % start of an individual contribution

\maketitle

\begin{abstract}
  Numerous popular abstract strategy games ranging from \gamename{hex} and \gamename{havannah} to \gamename{lines of action} belong to the class of connection games.
  Still, very few complexity results on such games have been obtained since \gamename{hex} was proved \pspace-complete in the early eighties.

  We study the complexity of two connection games among the most widely played.
  Namely, we prove that \gamename{havannah} and \gamename{twixt} are \pspace-complete.

  The proof for \gamename{havannah} involves a reduction from \gamename{generalized geography} and is based solely on ring-threats to represent the input graph.
  On the other hand, the reduction for \gamename{twixt} builds up on previous work as it is a straightforward encoding of \gamename{hex}.

\end{abstract}

\section{Introduction}
A connection game is a kind of abstract strategy game in which players try to make a specific type of connection with their pieces~\cite{Browne2005}.
In many connection games, the goal is to connect two opposite sides of a board.
In these games, players take turns placing or/and moving pieces until they connect the two sides of the board.
\gamename{Hex}, \gamename{twixt}, and \gamename{slither} are typical examples of this type of game.
However, a connection game can also involve completing a loop (\gamename{havannah}) or connecting all the pieces of a color (\gamename{lines of action}).

A typical process in studying an abstract strategy game, and in particular a connection game, is to develop an artificial player for it by adapting standard techniques from the game search literature, in particular the classical Alpha-Beta algorithm~\cite{Anshelevich2002} or the more recent Monte Carlo Tree Search paradigm~\cite{BrownePWLCRTPSC2012,ArnesonHH2010}.
These algorithms explore an exponentially large game tree are meaningful when optimal polynomial time algorithms are impossible or unlikely.
For instance, tree search algorithms would not be used for \gamename{nim} and \gamename{Shannon's edge switching game} which can be played optimally and solved in polynomial time~\cite{BrunoW1970}.

The complexity class \pspace{} comprizes those problems that can be solved on a Turing machine using an amount of space polynomial in the size of the input.
The prototypical example of a \pspace{}-complete problem is the Quantified Boolean Formula problem (\textsc{qbf}) which can be seen as a generalization of \textsc{sat} allowing for variables to be both existentially and universally quantified.
Proving that a game is \pspace{}-hard shows that a variety of intricate problems can be encoded via positions of this game.
Additionally, it is widely believed in complexity theory that if a problem is \pspace-hard, then it admits no polynomial time algorithms.
%Complexity results such as \pspace{}-hardness can thus be seen as a justification both that the game is rich and non trivial and that a tree search approach is relevant.

For this reason, studying the computational complexity of games is a popular research topic.
The complexity class of \gamename{chess} and \gamename{go} was determined shortly after the very definition of these classes and %~\cite{FraenkelL1981,Robson1983}.
other popular games have been classified since then~\cite{DemaineH2009,HearnDemaine2009}. %, including \gamename{gomoku} and \gamename{othello}~\cite{IwataK1994,FurtakKUB2005}.
More recently, we studied the complexity of trick taking card games which notably include \gamename{bridge}, \gamename{skat}, \gamename{tarot}, and \gamename{whist}~\cite{BonnetJamainSaffidine2013IJCAI}.
 
Connection games have received less attention.
Besides Even and Tarjan's proof that \gamename{Shannon's vertex switching game} is \pspace-complete~\cite{EvenTarjan1976} and Reisch's proof that \gamename{hex} is \pspace-complete~\cite{Reisch1981}, the only complexity results on connection games that we know of are the \pspace-completeness of virtual connection detection~\cite{Kiefer2003} in \gamename{hex}, the \np-completeness of dominated cell detection in \gamename{Shannon's vertex switching game}~\cite{BjornssonHJvR2007}, as well as an unpublished note showing that a problem related to \gamename{twixt} is \np-complete~\cite{MazzoniW1997}.\footnote{For a summary in English of Reisch's reduction, see Maarup's thesis~\cite{Maarup2005}.}
%Complexity results are rarely known, \gamename{hex} is the only one for which a result of \pspace{}-completeness is known.

The two games that we study in this paper rank among the most notable connection games.
They were the main topic of multiple master's theses and research articles~\cite{Huber1983,MazzoniW1997,Moesker2009,UiterwijkM2009,TeytaudT2010,Lorentz2011,Ewalds2012}, and they both gave rise to competitive play.
High-level online competitive play takes place on \url{www.littlegolem.net}.
Finally, live competitive play can also be observed between human players at the Mind Sports Olympiads where an international \gamename{twixt} championship has been organized every year since 1997, as well as between \gamename{havannah} computer players at the ICGA Computer Olympiad since 2009.\footnote{See \url{www.boardability.com/game.php?id=twixt} and \url{www.grappa.univ-lille3.fr/icga/game.php?id=37} for details.}

%\gamename{twixt} 
%We propose in this paper a study of the complexity of \gamename{havannah} and \gamename{twixt}.

%maybe préciser le problème de décision associé

%baratin sur ce qui existe sur les autres classes de jeux, territoire, capture ... pkoi rien sur les jeux de connexions hormi hex ? pcq en fait ya des trucs mais bien cachés.

%les jeux de connexions ne sont pas forcément tous dans PSPACE, lines of action me semble justement être un jeu très intéressant de ce point de vu. J'ai pas vraiment d'a priori sur ce que ça peut être. Il faudra s'y mettre.
%Breakthrough et EWN ne sont et ne seront jamais des jeux de connexions :P

%parler du complexity lanscape extended to connection games (trouver la signification de ceci fait parti du jeu)

%We present in Section 2 our main result, the \pspace{}-completeness of \gamename{havannah}.
%Then in Section 3 and 4 the \pspace{}-completeness of \gamename{twixt} and \gamename{slither}.

%Faire une abdallah pour rendre le tout nice.

\section{Havannah}
%Complexity Landscape For The Win!
%Some useful words: quintessential, stark, noteworthy, mischievous.
\gamename{havannah} is a 2-player connection game played on a hexagonal board paved by hexagons.
White and Black place a stone of their color in turn in an unoccupied cell.
Stones cannot be taken, moved nor removed.
Two cells are neighbors if they share an edge.
A group is a connected component of stones of the same color via the neighbor relation.
A player wins if they realize one of the three following different structures: a circular group, called \emph{ring}, with at least one cell, possibly empty, inside; a group linking two corners of the board, called \emph{bridge}; or a group linking three edges of the board, called \emph{fork}.

As the length of a game of \gamename{havannah} is polynomially bounded, exploring the whole game tree can be done with polynomial space, so \Hav{} is in \pspace.

In our reduction, the \gamename{havannah} board is large enough that the gadgets are far from the edges and the corners.
Additionally, the gadgets feature ring threats that are short enough that the bridges and forks winning conditions do not have any influence.
Before starting the reduction, we define threats and make two observations that will prove useful in the course of the reduction.

A \emph{simple threat} is defined as a move which threatens to realize a ring on the next move on a unique cell.
There are only two kinds of answers to a simple threat: either win on the spot or defend by placing a stone in the cell creating this very threat.
A \emph{double threat} is defined as a move which threatens to realize a ring on the next move on at least two different cells.
We will use \emph{threat} as a generic term to encompass both simple and double threats.
A \emph{winning sequence of threats} is defined as a sequence of simple threats ended by a double threat for one player such that the opponent's forced move never makes a threat.
Thus, when a player is not threatened and can initiate a winning sequence of threats, they do win.
To be more concise, we will denote by $W:a_1$,$a_2$; $a_3$,$a_4$; \dots; $a_{2n-1}$(,$a_{2n}$) the sequence of moves starting with White's move $a_1$, Black's answer $a_2$, and so on.
$a_{2n}$ is optional, for the last move of the sequence might be White's or Black's.
Similarly, $B:a_1$,$a_2$; $a_3$,$a_4$; \dots; $a_{2n-1}$(,$a_{2n}$) denotes the corresponding sequence of moves initiating by Black.
We will use the following lemmas multiple times:
\begin{lemma}
  \label{lem:simple-threat}
  If a player is not threatened, playing a simple threat forces the opponent to answer on the cell of the threat.
\end{lemma}
%\begin{withproofs}
\begin{proof}
  Otherwise, no matter what have played their opponent, 
  placing a stone on the cell of the threat wins the game.
  \qed
\end{proof}
%\end{withproofs}
\begin{lemma}
\label{lem:double-threat}
If a player is not threatened, playing a double threat is winning.
\end{lemma}
%\begin{withproofs}
\begin{proof}
 The player is not threatened, so their opponent can not win at their turn.
 Let $u$ and $v$ be two cells of the double threat.
 If their opponent plays in $u$, the player wins by playing in $v$.
 If their opponent plays somewhere else, the player wins by playing in $u$.
 \qed
\end{proof}
%\end{withproofs}

\subsection{Generalized Geography.}
\label{sec:gg}
\gamename{Generalized geography} (\GG{}) is one of the first two-player games to have been proved \pspace-complete~\cite{Schaefer1978}.
It has been used to reduce to multiple games including \gamename{hex}, \gamename{othello}, and \gamename{amazons}~\cite{Reisch1981,FurtakKUB2005,DemaineH2009}.

In \GG{}, players take turns moving a token from vertex to vertex.
If the token is on a vertex $v$, then it can be moved to a vertex $v'$ neighboring $v$ provided $v'$ has not been visited yet.
A player wins when it is their opponent's turn and the oppoent has no legal moves.
An instance of \GG{} is a graph $G$ and an initial vertex $v_0$, and asks whether the first player has a winning strategy in the corresponding game.

We denote by $P(v)$ the set of predecessors of the vertex $v$ in $G$, and $S(v)$ the set of successors of $v$.
A vertex with in-degree $i$ and out-degree $o$ is called $(i,o)$-vertex.
The degree of a vertex is the sum of the in-degree and the out-degree, and the degree of $G$ is the maximal degree among all vertices of $G$.
If $V$ is the set of vertices of $G$ and $V'$ is a subset of vertices, then $G[V\setminus V']$ is the induced subgraph of $G$ where vertices belonging to $V'$ have been removed.

Lichtenstein and Sipser have proved that the game remained \pspace-hard even if $G$ was assumed to be bipartite and of degree at most 3~\cite{LichtensteinS1980}.
We will reduce from such a restriction of \GG{} to show that \Hav{} is \pspace-hard.
To limit the number of gadgets we need to create, we will also assume a few simplifications detailed below.
An example of a simplified instance of \GG{} can be found in Fig.~\ref{fig:gg-instance}.

\begin{figure}
  \centering
  \includegraphics{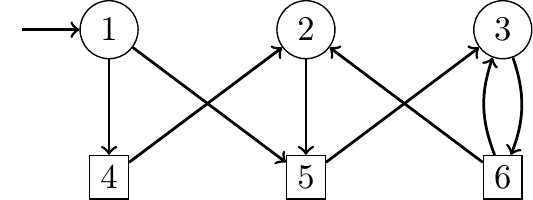}
  \caption{Example of an instance of \GG{} with vertex $1$ as initial vertex.}
  \label{fig:gg-instance}
\end{figure}

Let $(G,v_0)$ be an instance of \GG{} with $G$ bipartite and of degree at most 3.
We can assume that there is no vertex $v$ with out-degree 0 in $G$.
Indeed, if $v_0 \in P(v)$ then $(G,v_0)$ is trivially winning for Player 1.
Else, $(G[V \setminus (\{v\} \cup P(v))],v_0)$ is an equivalent instance, since playing in a predecessor of $v$ is losing.

All edges coming to the initial vertex $v_0$ can be removed to form an equivalent instance.
So, $v_0$ is a $(0,1)$-, a $(0,2)$-, or a $(0,3)$-vertex.
If $S(v_0)=\{v'\}$, then $(G[V \setminus \{v_0\}],v')$ is a strictly smaller instance such that Player 1 is winning in $(G,v_0)$ if and only if Player 1 is losing in $(G[V \setminus \{v_0\}],v')$.
If $S(v_0)=\{v',v'',v'''\}$, then Player 1 is winning in $(G,v_0)$ if and only if Player 1 is losing in at least one of the three instances $(G[V \setminus \{v_0\}],v')$, $(G[V \setminus \{v_0\}],v'')$, and $(G[V \setminus \{v_0\}],v''')$.
In those three instances $v'$, $v''$, and $v''$ are not $(0,3)$-vertices since they had in-degree at least 1 in $G$.
Therefore, we can also assume that $v_0$ is $(0,2)$-vertex.

We call an instance with an initial $(0,2)$-vertex and then only $(1,1)$-, $(1,2)$-, and $(2,1)$-vertices a \emph{simplified} instance.

In the following subsections we propose gadgets that encode the different parts of a \emph{simplified} instance of \GG{}.
These gadgets have starting points and ending points.
The gadgets are assembled so that the ending point of a gadget coincides with the starting point of the next one.
The resulting instance of \gamename{havannah} is such that both players must enter in the gadgets by a starting point and leave it by an ending point otherwise they lose.

\subsection{Edge gadgets.}
Wires, curves, and crossroads will enable us to encode the edges of the input graph.
In the representation of the gadgets, White and Black stones form the proper gadget.
Dashed stones and gray stones are respectively White and Black stones setting the context.

In the \gamename{havannah} board we name the 6 directions: North, North-West, 
South-West, South, South-East, and North-East according to standard designation.
While figures and lemmas are mostly presented from White's point of view, all the gadgets and lemmas work exactly the same way with colors reversed.
%Non Abdallah, n'efface pas ça! te-bâtir! 

\paragraph{The wire gadget.}
Basically, a wire teleports moves: one player plays in a cell $u$ and their opponent has to answer in a possibly remote cell $v$.
$u$ is called the starting point of the wire and $v$ is called its ending point.
A wire where White prepares a threat and Black answers is called a WB-wire (Fig.~\ref{fig:hav-entire-wire}); conversely, we also have BW-wires.
We say that WB-wires and BW-wires are \emph{opposite} wires.
Note that wires can be of arbitrary length and can be curved with 120$^\circ$ angles (Fig.~\ref{fig:hav-curved-wire}).
On an empty board, a wire can link any pair of cells as starting and ending point provided they are far enough from each other.
\begin{figure}
  \centering
  \subfloat[Entire WB-wire which starts in $u$ and ends in $v$.]{
    \centering
    \includegraphics{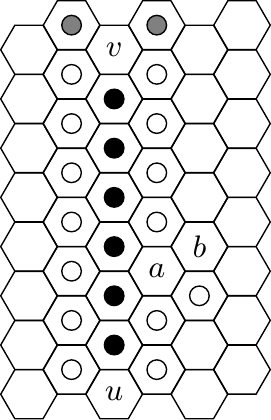}
%    \hspace{5mm}
    \label{fig:hav-entire-wire}
  }
  \hfill
  \subfloat[Curved fragment for a BW-wire.]{
    \centering
    \includegraphics{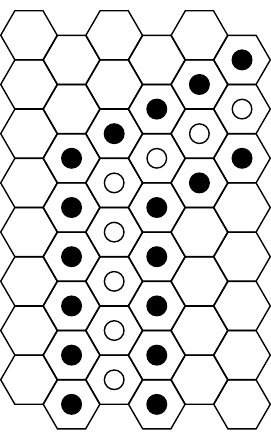}
    \label{fig:hav-curved-wire}
  }
  \hfill
  \subfloat[Crossover gadget.]{
    \centering
    \includegraphics{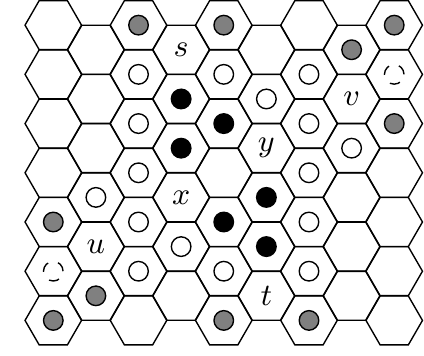}
    \label{fig:hav-crossroad}
  }
  \caption{Edge gadgets.}
  \label{fig:hav-wire}
\end{figure}

\begin{lemma}
\label{lem:hav-wire}
  If White plays in the starting point $u$ of a WB-wire (Fig.~\ref{fig:hav-entire-wire}), and Black does not answer by a threat, Black is forced to play in the ending point $v$ (possibly with moves at $a$ and $b$ interleaved).
\end{lemma}
%\begin{withproofs}
\begin{proof}
   If Black does not play neither in $a$ nor in $b$, then White plays in $v$ which makes a double threat in $a$ and $b$ and wins by Lemma~\ref{lem:double-threat}.
If Black plays in $a$ (resp.~in $b$), at the very least White can play in $b$ (resp.~in $a$) which forces Black to play in $v$ by Lemma~\ref{lem:simple-threat}.
 %  If Black plays in $a$ or $b$, then White plays in $v$.
 %  By Lemma~\ref{lem:simple-threat}, Black is forced to play in $a$ or $b$ (depending on what he has already played at the preceding move).
 %  Finally, White plays in $e$ which constitutes a double threat in $c$ and $d$ and wins.
   \qed
\end{proof}
%\end{withproofs}

\paragraph{The crossover gadget.} 
The input graph of \GG{} might not be planar, so we have to design a crossover gadget to enable two chains of wires to cross.
Fig.~\ref{fig:hav-crossroad} displays a crossover gadget, we have a South-West BW-wire with starting point $u$ which is linked to a North-East BW-wire with ending point $v$, and a North BW-wire with starting point $s$ is linked to a South BW-wire with ending point $t$.

\begin{lemma}
  \label{lem:hav-crossroad}
  In a crossover gadget (Fig.~\ref{fig:hav-crossroad}), if White plays in the starting point $u$, Black ends up playing in the ending point $v$ and if White plays in the starting point $s$, Black ends up playing in the ending point $t$.
\end{lemma}
%\begin{withproofs}
\begin{proof}
  By Lemma~\ref{lem:simple-threat}, if White plays in $u$, Black has to play in $x$, forcing White to play in $y$, forcing finally Black to play in $v$.
  If White plays in $s$, again by Lemma~\ref{lem:simple-threat}, Black has to play in $t$.
  \qed
\end{proof}
%\end{withproofs}

Note that the South wire is linked to the North wire irrespective of whether the other pair of wires has been used and conversely.
That is, in a crossover gadget two paths are completely independent.
%That is, having if the \emph{South} chain or the \emph{North-East} chain of wires have already been played into does not perturbate the other chain.

\subsection{Vertex gadgets.}
We now describe the gadgets encoding the vertices.
Recall from Section~\ref{sec:gg} that simplified \GG{} instances only feature $(1,2)$-, $(1,1)$-, and $(2,1)$-vertices, and a $(0,2)$-vertex.
One can encode a $(1,1)$-vertex with two consecutive opposite wires.
Thus, we will only present three vertex gadgets, one for $(2,1)$-vertices, one for $(1,2)$-vertices, and one for the $(0,2)$-vertex.

\paragraph{The (2,1)-vertex gadget.}

A $(2,1)$-vertex gadget receives two wire ending points.
If a stone is played on either of those ending points, it should force an answer in the starting point of a third wire.
That simulates a vertex with two edges going in and one edge going out.

\begin{figure}
  \centering
  \subfloat[Gadget before being used.
    The wires for the in-edges end in $u$ and $v$, the wire for the out-edge starts in $s$.]{
    \centering
    \includegraphics{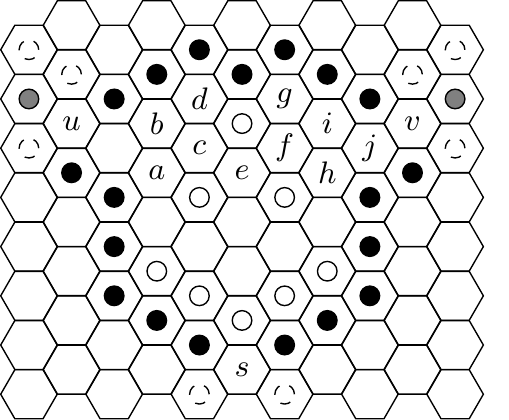}
    \label{fig:hav-21}
  }
  \hfill
  \subfloat[Gadget after being used and then reentered.
    White wins with a double threat.
  ]{
    \centering
    \includegraphics{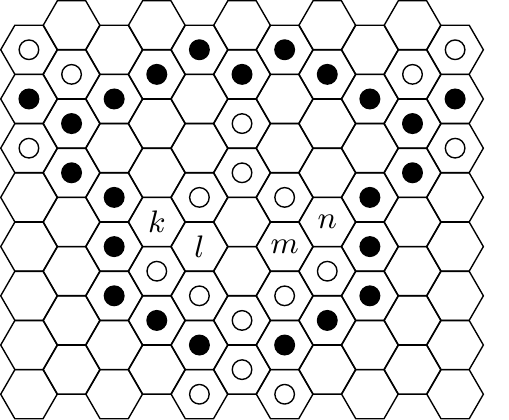}
    \label{fig:hav-reentering-21}
  }
  \caption{The $(2,1)$-vertex gadget links three WB-wires.
    The North-West and North-East ones end in $u$ and $v$, the South WB-wire starts in $s$.
  }
\end{figure}

\begin{lemma}\label{lem:hav-21}
  If Black plays in one of the two possible starting points $u$ and $v$ of a $(2,1)$-vertex gadget (Fig.~\ref{fig:hav-reentering-21}), and White does not answer by a threat, White is forced to play in the ending point $s$.
\end{lemma}
%\begin{withproofs}
\begin{proof}
  Assume Black plays in $u$ and White answers by a move which is not in $s$ nor a threat.
  This move from White has to be either in $v$ or in $j$, otherwise, Black has a double threat by playing in $s$ and wins by Lemma~\ref{lem:double-threat}.
  Suppose White plays in $v$.
  Now, Black plays in $s$ with a simple threat in $j$, so White has to play in $j$ by Lemma~\ref{lem:simple-threat}.
  Then Black has the following winning sequence: B: $a$,$b$; $c$,$d$; $h$,$i$; $f$. Black has now a double threat in $g$ and $e$ and so wins by Lemma~\ref{lem:double-threat}.
  If White plays in $j$ instead of $v$, the argument is similar.

  If Black plays the first move in $v$, the proof that White has to play in $s$ is similar.
  \qed
\end{proof}
%\end{withproofs}

\paragraph{The (1-2)-vertex and (0,2)-vertex gadgets.}
  
A $(1,2)$-vertex gadget receives one ending point of a wire (Fig.~\ref{fig:hav-12}).
If a stone is played on this ending point, it should offer the choice to defend either by playing in the starting point of a second wire, or by playing in the starting point of a third wire.
That simulates a vertex with one edge going in and two edges going out.
The $(0,2)$-vertex gadget (or \emph{starting-vertex} gadget) can be seen as a $(1,2)$-vertex gadget where a stone has already been played on the ending point of the in-edge. The $(0,2)$-vertex gadget represents the two possible choices of the first player at the beginning of the game.

\begin{figure}
  \centering
  \subfloat[The $(1,2)$-vertex gadget.
%      A North WB-wire ending in $u$ corresponds to the incoming edge.
    The wire for the in-edge ends in $u$.
  ]{
    \centering
    \hspace{1cm}
    \includegraphics{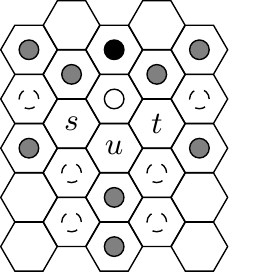}
    \hspace{5mm}
    \label{fig:hav-12}
  }
  \hfill
  \subfloat[The $(0,2)$-vertex gadget representing the starting vertex $v_0$.]{
    \centering
    \hspace{1cm}
    \includegraphics{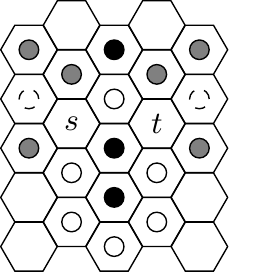}
    \hspace{1cm}
    \label{fig:hav-starting-vertex}
  }
  \caption{In these choice gadgets, White can defend by playing in $s$ or in $t$.
    A North-West BW-wire starts in $s$ and a North-East BW-wire starts in $t$.
  }
  \label{fig:hav-choice-vertices}
\end{figure}

\begin{lemma}
  \label{lem:hav-12}
  If Black plays in the starting point $u$ of a $(1,2)$-vertex gadget (Fig.~\ref{fig:hav-12}), and White does not play a threat, White is forced to play in one of the two ending points $s$ and $t$.
  Then, if Black does not answer by a threat, they have to play in the other ending point.
\end{lemma}
%\begin{withproofs}
\begin{proof}
  Black plays in $u$.
  Suppose White plays neither in $s$ nor in $t$ nor a threatening move.
  Then Black plays in $s$.
  By Lemma~\ref{lem:simple-threat}, White has to play in $t$ but Black wins by playing in the ending point of the wire starting at $s$ by Lemma~\ref{lem:hav-wire}.

  Assume White's answer to $u$ is to play in $s$.
  $t$ can now be seen as the ending point of the in-wire, so Black needs to play in $t$ or make a threat by Lemma~\ref{lem:hav-wire}.
  \qed
\end{proof}
%\end{withproofs}

\begin{corollary}
 \label{lem:hav-starting-vertex}
  If White is forced to play a threat or to open the game in one of the two opening points $s$ and $t$ of the $(0,2)$-vertex gadget (Fig.~\ref{fig:hav-starting-vertex}).
  Then, if Black does not play a threat, they are forced to play in the other opening point.
\end{corollary}

\subsection{Assembling the gadgets together.}

Let $(G,v_0)$ be a simplified instance of \GG{}, and $n$ be its number of vertices.
$G$ being bipartite, we denote by $V_1$ the side of the partition containing $v_0$, and $V_2$ the other side.
Player 1 moves the token from vertices of $V_1$ to vertices of $V_2$ and player 2 moves the token from $V_2$ to $V_1$.
We denote by $\phi$ the reduction from \GG{} to \Hav{}.
Let us describe the construction of $\phi((G,v_0))$.
As an example, we provide the reduction from the \GG{} instance from Fig.~\ref{fig:gg-instance} in Fig.~\ref{fig:hav-reduction}.

\begin{figure}
  \centering
  \includegraphics{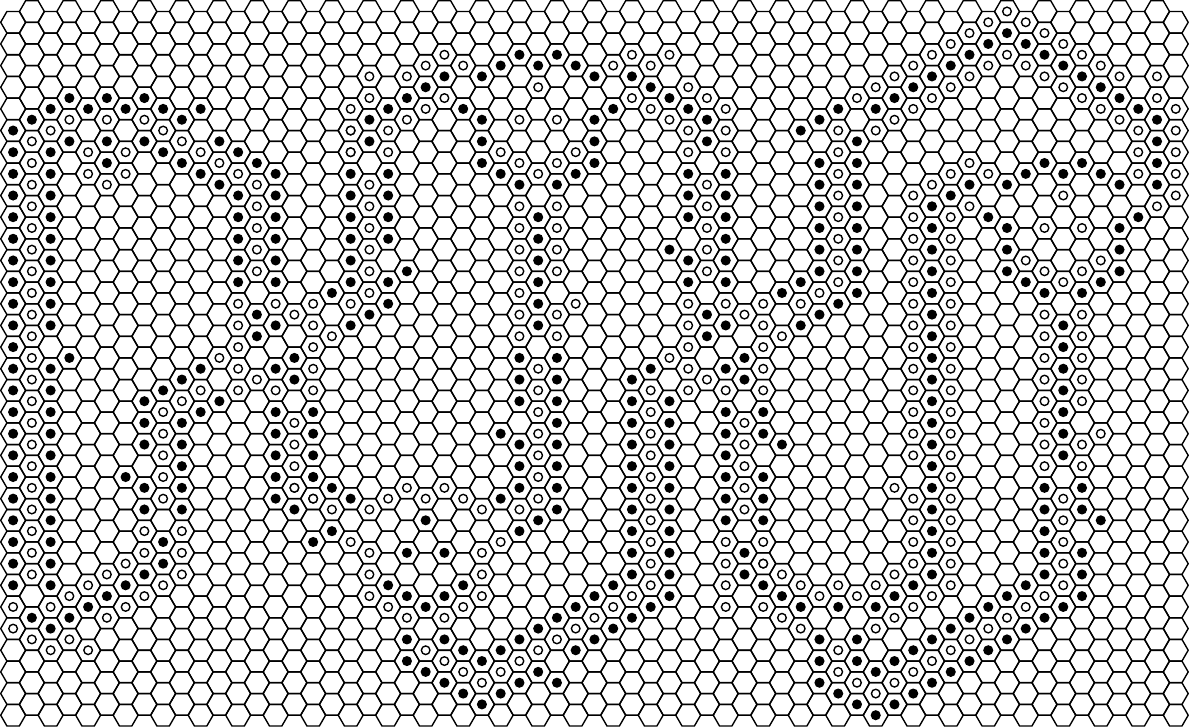}
  \caption{\gamename{havannah} gadgets representing the \GG{} instance from Fig.~\ref{fig:gg-instance}.}
  \label{fig:hav-reduction}
\end{figure}

The initial vertex $v_0$ is encoded by the gadget displayed in Fig.~\ref{fig:hav-starting-vertex}.
Each player 1's $(2,1)$-vertex is encoded by the $(2,1)$-vertex gadget of Fig.~\ref{fig:hav-21}, and each player 2's $(2,1)$-vertex is encoded by the same gadget in reverse color.
Each player 1's $(1,2)$-vertex is encoded by the $(1,2)$-vertex gadget of Fig.~\ref{fig:hav-12}, and each player 2's $(1,2)$-vertex is encoded by the same gadget in reverse color.

All White's vertex gadgets are aligned and all Black's vertex gadgets are aligned on a parallel line.
%Let $\alpha$ be the size of the smallest wire and $\beta$ be the size of a crossover.
%Then $2\alpha+\beta n$ rows separate Black's line from White's line.
%We connect with wires and crossroads the exit of a gadget encoding a vertex $u$ to one of the entrance of a gadget encoding $v$ whenever $(u,v)$ is an edge in $G$.
Whenever $(u, v)$ is an edge in $G$, we connect an exit of the vertex gadget representing $u$ to an entrance of a gadget encoding $v$ using wires and crossover gadgets.
Let $n$ be the number of vertices in $G$, since $G$ is of degree 3, we know that the number of edges is at most $\nicefrac{3n}{2}$.
The minimal size in terms of \gamename{havannah} cells for a smallest wire and the size of a crossover are constants.
Therefore the distance between Black's line and White's line is linear in $n$.
Note that, two wires of opposite colors might be needed to connect two vertex gadgets or a vertex gadget and a crossover.
%For that, one wire, up to $n$ crossover gadgets, and up to two opposite wires are sufficient, hence the distance $2 \alpha+\beta n$.
Similarly, we can show that the distance between two vertices on Black's line or on White's line is constant. %and separated by $\gamma$, where $\gamma$ is a constant which leaves enough room for two wires.

%\begin{lemma}
 % \label{lem:hav-starting-vertex}
  %If White does not play a threat, they are forced to play in $s$ or in $t$.
  %Then, if Black does not play a threat, they are forced to play in $t$ or in $s$ respectively.
%\end{lemma}
%\begin{withproofs}
%\begin{proof}
 % See the proof of Lemma~\ref{lem:hav-12}.
%\end{proof}
%\end{withproofs}

\begin{lemma}
  \label{lem:hav-reentering}
  If Black reenters a White's $(2,1)$-vertex gadget (Fig.~\ref{fig:hav-reentering-21}), and Black has no winning sequence of threats elsewhere, White wins.
\end{lemma}
%\begin{withproofs}
\begin{proof}
  If Black reenters a White's $(2,1)$-vertex by playing in $v$, White plays in $e$.
%  This move defends the winning sequence of threats for Black: $B:h$,$i$; $f$,$g$; $a$,$b$; $c$,$d$; $e$.
  As Black cannot initiate a winning sequence, whatever he plays White can defend until Black is not threatening anymore.
  Then White plays in $k$ or in $l$ with a decisive double threat in $m$ and $n$.
  \qed
\end{proof}
%\end{withproofs}

\begin{theorem}
  \label{hav-main-th}
  \Hav{} is \pspace{}-complete.
\end{theorem}
%\begin{withproofs}
\begin{proof}
  We already mention that \Hav{} $\in$ \pspace{} and we just presented a polynomial time reduction from a \pspace-complete problem.
  We shall now prove that the reduction is sound, that is: player 1 is winning in $(G,v_0)$ if and only if White is winning in $\phi((G,v_0))$.
  First we show that the players in the game of \gamename{havannah} lose if they make a move which does not correspond to anything in the instance of \GG{}.
  Such a move will be called a \emph{cheating} move.
  The exhaustive list of non cheating moves is: defending a threat, playing at the end of a wire when the opponent had just play at its starting point, choosing which wire starting point $s$ or $t$ to block when the opponent had just play in $u$ (Fig.~\ref{fig:hav-12}), which forces them to take the other wire, and playing in the exit $s$ of a $(2,1)$-vertex gadget when the opponent had just play in $u$ or in $v$ (Fig.~\ref{fig:hav-21}).
  In order to conclude by invoking Lemma~\ref{lem:hav-wire} up to Corollary~\ref{lem:hav-starting-vertex}, we should show that making a threat is not helpful in all the above situations.
  Note that those Lemmas imply the following invariant: while White and Black play a legal game of \GG{}, at their turn, a player is threatened or their opponent can initiate a winning sequence of threats.
  There is only two kinds of places where one can play a threat: the crossroad gadget (Fig.~\ref{fig:hav-crossroad}) and the $(2,1)$-vertex gadget while already being entered (Fig.~\ref{fig:hav-reentering-21}).
  
  Let us start by showing that playing a threat in a crossroad gadget which does not defend a threat, that is, the action was occurring in a different place, is losing.
  If White plays in $s$ then Black plays in $t$ which is the starting point of a BW wire.
  And now, they are at least two places where Black can initiate a winning sequence of threats, so White loses (after possibly playing some additional but harmless threats).
  The same holds by reversing the colors or by reversing $s$ and $t$, and is not affected by whether or not stones have been played in $u$, $x$, $y$ and $v$.
  If, instead, White plays in $u$, Black answers in $x$, White answers in $y$ and Black plays in $v$, and again Black can initiate a winning sequence of threats in two places.
  If, instead, White plays in $x$, Black answers in $u$ and again White is losing.
  If, instead, Black plays in $x$, White plays in $y$ and Black plays in $v$, and now White plays their winning sequence of threats.
  
  Now, let us show that the threats in the already entered $(2,1)$-vertex gadget are harmless.
  Consider now Fig.~\ref{fig:hav-reentering-21}.
  If Black plays in $b$, White answers in $a$ and there is no more threats for Black.
  If Black plays in $a$, White answers in $b$.
  Black can threat again in $c$ or $d$ but White defends in $d$ or $c$, respectively, and there are no more threats.
  Note that this does not affect the fact that reentering in the gadget is losing for Black.
  Summing up, White and Black has to simulate a proper game of \GG{} in the instance $(G,v_0)$.
  
  We now show that if a player in the game of \gamename{havannah} has no more move in the corresponding \GG{} instance, they lose.
  The only non cheating move would be to reenter in a $(2,1)$-vertex but it is losing by Lemma~\ref{lem:hav-reentering}.
  \qed
\end{proof}
%\end{withproofs}
%\input{gg-havannah}
\section{TwixT}
%We show in this section that \gamename{twixt} is \pspace{}-complete.

Alex Randolph's \gamename{twixt} is one of the most popular connection games.
It was invented around 1960 and was marketed as soon as in 1962~\cite{Huber1983}.
In his book devoted to connection games, Cameron Browne describes \gamename{twixt} as one of the most popular and widely marketed of all connection games~\cite{Browne2005}.
We now briefly describe the rules of \gamename{twixt} and refer to Moesker's master's thesis for an introduction and a mathematical approach to the strategy, and the description of a possible implementation~\cite{Moesker2009}.

\gamename{twixt} is a 2-player connection game played on a \gamename{go}-like board.
At their turn, player White and Black place a pawn of their color in an unoccupied place.
Just as in \gamename{havannah} and \gamename{hex}, pawns cannot be taken, moved, nor removed.
When 2 pawns of the same color are spaced by a knight's move, they are linked by an edge of their color, unless this edge would cross another edge.
At each turn, a player can remove some of their edges to allow for new links.
The goal for player White (resp.~Black) is to link top and bottom (resp.~left and right) sides of the board.
Note that sometimes, a player could have to choose between two possible edges that intersect each other.
The \textit{pencil and paper} version \gamename{twixtpp} where the edges of a same color are allowed to cross is also famous and played online.

As the length of a game of \gamename{twixt} is polynomially bounded, exploring the whole tree can be done with polynomial space using a minimax algorithm.
Therefore \gamename{twixt} is in \pspace.

Mazzoni and Watkins have shown that \textsc{3-sat} could be reduced to single-player \gamename{twixt}, thus showing \np-completeness of the variant~\cite{MazzoniW1997}.
While it might be possible to try and adapt their work and obtain a reduction from \textsc{3-qbf} to standard two-player \gamename{twixt}, we propose a simpler approach based on \gamename{hex}. %\footnote{\textsc{3-qbf} is a \pspace-complete generalization of \textsc{3-sat} to allow quantifiers.}
The \pspace-completeness of \gamename{hex} has already been used to show the \pspace{}-completeness of \gamename{amazons}, a well-known territory game~\cite{FurtakKUB2005}.

We now present how we construct from an instance $G$ of \gamename{hex} an instance $\phi(G)$ of \gamename{twixt}.
We can represent a cell of \gamename{hex} by the $9\times 9$ \gamename{twixt} gadgets displayed in Fig.~\ref{fig:twixt-cell}.
Let $n$ be the size of a side of $G$, Fig.~\ref{fig:twixt-board} shows how a \gamename{twixt} board can be paved by $n^2$ \gamename{twixt} cell gadgets to create a \gamename{hex} board.

\begin{figure}
  \centering
  \subfloat[Empty cell.]{
    \centering
    \includegraphics{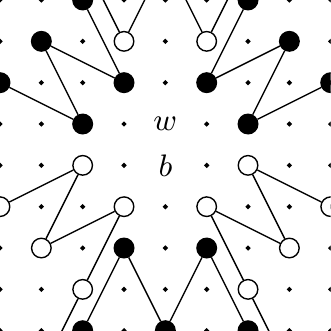}
    \label{fig:tw-empty}
  }
  \hfill
  \subfloat[White cell.]{
    \centering
    \includegraphics{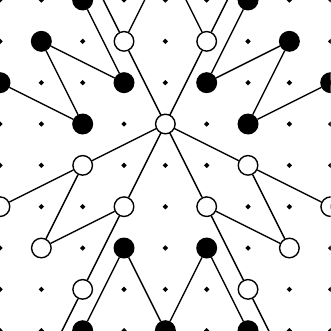}
    \label{fig:tw-white}
  }
  \hfill
  \subfloat[Black cell.]{
    \centering
    \includegraphics{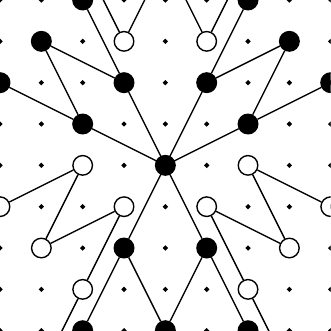}
    \label{fig:tw-black}
  }
  \caption{Basic gadgets needed to represent cells.}
  \label{fig:twixt-cell}
\end{figure}
 
It is not hard to see from Fig.~\ref{fig:tw-empty} that in each gadget of Fig.~\ref{fig:twixt-board}, move $w$ (resp.~$b$) is dominating for White (resp.~Black).
That is, playing $w$ is as good for White as any other move of the gadget.
We can also see that the moves that are not part of any gadget in Fig.~\ref{fig:twixt-board} are dominated for both players.
As a result, if player Black (resp.~White) has a winning strategy in $G$, then player Black has a winning strategy in $\phi(G)$.
Thus, $G$ is won by Black if and only if $\phi(G)$ is won by Black.
Therefore determining the winner in \gamename{twixt} is at least as hard as in \gamename{Hex}, leading to the desired result.

\begin{figure}[h!]
  \centering
  \includegraphics{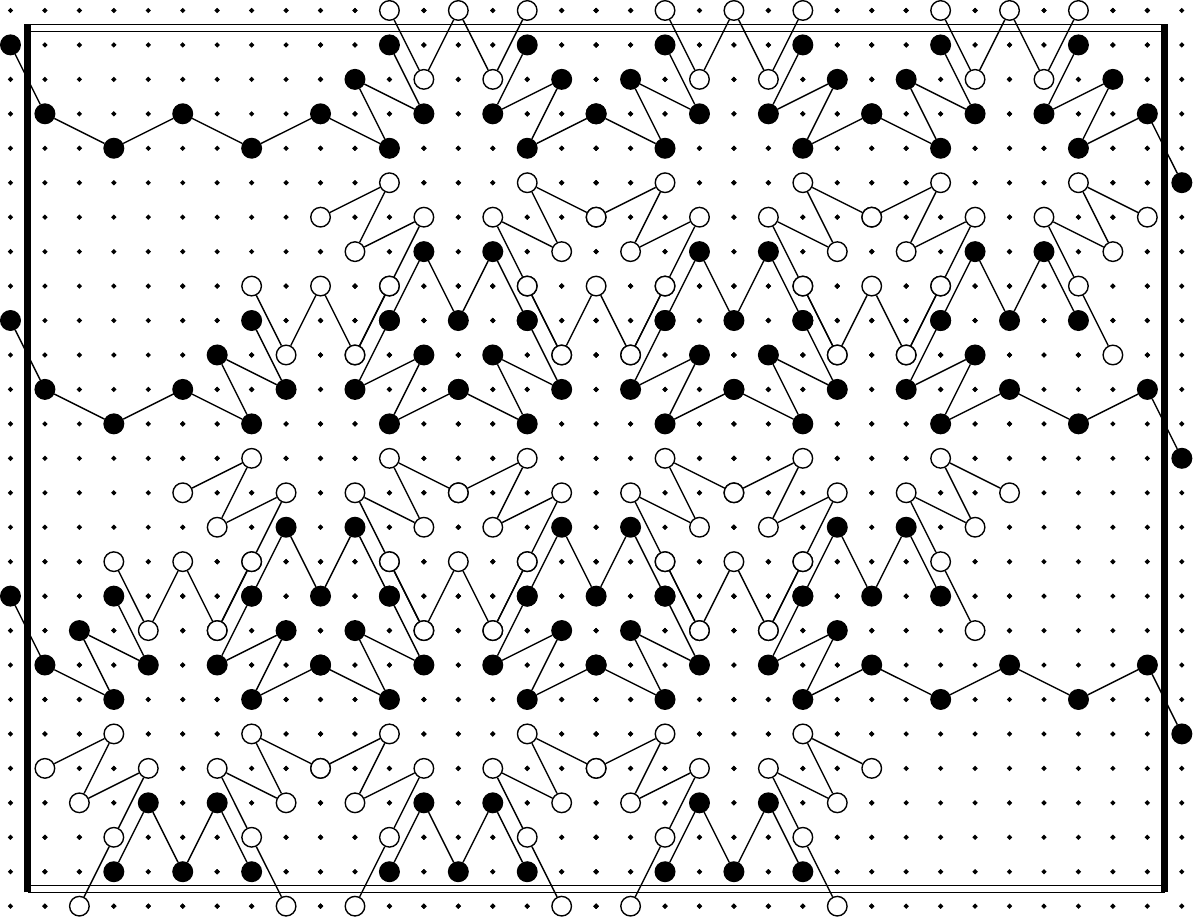}
  \caption{Empty $3\times 3$ \gamename{hex} board reduced to a \gamename{twixt} board.}
  \label{fig:twixt-board}
\end{figure}

%\begin{lemma}
 % Any move on the edge of a cell is irrelevant for that cell.

  %Therefore, if Black has a winning \gamename{hex} strategy, they can apply it in the \gamename{twixt} reduction and win.
%\end{lemma}

%Using the last gadgets, we can represent the $n^2$ squares and their neighboring of an instance of \Hex{}.
%Then, we connect all the gadgets on the left (resp.~top, right, bottom) sides to physical corresponding sides in \Tw{}.
%Finally, except the two places in each gadget, we can complete the full board by a pawn of color white or black such %that the number of pawns of each color is the same.
%Thus, we obtain an instance of \Tw{} with only $2n^2$ unoccupied places.

%Using this reduction, let's index the $n^2$ gadgets and denote by $w_i$ (resp.~$b_i$) for $i=1,\ldots,n^2$ the interesting place (see Figure~\ref{fig:tw-empty}) in gadget $i$ for player White (resp.~Black).

%% \begin{lemma}
%%   If player White (resp.~Black) has a winning strategy in $G$, then player White (resp.~Black) has a winning strategy in $\phi(G)$.
%% \end{lemma}
%% \begin{withproofs}
%% \begin{proof}
%%   It is clear since player White (resp.~Black) can replicate his strategy of $G$ in $\phi(G)$.
%%   When he has to play in square $i$ in $G$ he plays in place $w_i$ (resp.~$b_i$) in \gamename{twixt}.
%%   Moreover, player Black (resp.~White) has no interest to play other than $b_j$ for a $j \in \{1,\ldots,n^2\}$ since all the other moves are clearly irrelevant.
%% \end{proof}
%% \end{withproofs}

\begin{theorem}
  \gamename{twixt} is \pspace{}-complete.
\end{theorem}
\begin{withproofs}
\begin{proof}
  We already mentioned that \gamename{twixt} $\in$ \pspace.
  We presented a polynomial time reduction from a \pspace-complete problem.
  We shall prove that the reduction is sound.
  \qed
\end{proof}
\end{withproofs}

Observe that the proposed reduction holds both for the classic version of \gamename{twixt} as well as for the \emph{pencil and paper} version \gamename{twixtpp}.
Indeed, the reduction does not require the losing player to remove any edge, so it also proves that \gamename{twixtpp} is \pspace-hard.

%\section{Slither}
%\input{slither}

\section{Conclusion}

This paper establishes the \pspace{}-completeness of two important connection games, \gamename{havannah} and \gamename{twixt}.
The proof for \gamename{twixt} is a reduction from \gamename{hex} and applied to \gamename{twixtpp}.
The proof for \gamename{havannah} is more involved and is based on the Generalized Georgraphy problem restricted to bipartite graphs of degree 3.
This \gamename{havannah} reduction only used the loop winning condition, but it is easy to show that \gamename{havannah} without the loop winning condition can simulate \gamename{hex} and is \pspace{}-hard as well.
For both reductions, the size of the resulting game is only linearly larger than the size of the input instance.

The complexity of other notable connection games remains open.
In particular, the following games seem to be good candidates for future work on the complexity of connection games.

In \gamename{lines of action}, each player starts with two groups of pieces and tries to connect all their pieces by moving these pieces and possibly capturing opponent pieces~\cite{WinandsBS2010}.
While the goal of \gamename{lines of action} clearly makes it a connection game, the mechanics distinguishes it from more classical connection games as no pieces are added to the board and existing pieces can be moved or removed.
As a result, it is not even clear that \gamename{lines of action} is in \pspace{}.

\gamename{slither} is closer to \gamename{hex} but each move actually consists of putting a new stone on the board and possibly moving another one.
Obtaining a \pspace{}-hardness result for \gamename{slither} is not so easy since the rules allow a player to influence two different areas of the board in a single turn.
%A reduction must consider that place a stone in a ``gadget'' and move a stone in another one.

%The complexity of games, in particular the complexity of connection games, is usually studied with respect to the size of the board.
%For connection games the extension of the game to larger board is natural.
%Additionaly, it could be interesting to study the complexity relatively to the number of pawns available to each player.
%AND SO OBTAIN A COMPLEXITY LANDSCAPE !!!

\bibliographystyle{plain}

\end{document}